\documentclass{article}

\clearpage{}\usepackage{algorithm2e}
\usepackage{bbm}
\usepackage{mathtools}
\usepackage{bm}
\usepackage{tikz}
\usepackage{amsthm}
\usepackage{amssymb}
\usepackage{pifont}
\usepackage[margin=1in]{geometry}

\newcommand{\eps}{\varepsilon}
\newcommand{\reals}{\mathbb{R}}

\newcommand{\E}{\textbf{E}}
\newcommand\Prob{\textbf{Pr}}

\newcommand{\xmark}{\ding{55}}
\newcommand{\A}{{A}}
\newcommand{\Wel}{\textsc{Wel}}

\newcommand\Dist{\mathcal{D}}
\newcommand\Alloc{\mathcal{F}}
\newcommand\Mech{\mathcal{M}}
\newcommand\Domain{\mathcal{X}}
\newcommand\Range{\mathcal{Y}}
\newcommand\Exp{\mathbb{E}}

\newtheorem{theorem}{Theorem}

\newtheorem{remark}{Remark}

\newtheorem{lemma}{Lemma}

\newtheorem{definition}{Definition}

\clearpage{}

\begin{document}

\title{The Complexity of Black-Box Mechanism Design with Priors}

\author {
Evangelia Gergatsouli \\
UW-Madison\\
\tt{gergatsouli@wisc.edu}
\and
Brendan Lucier\\
Microsoft Research \\
\tt{brlucier@microsoft.com}
\and
Christos Tzamos\\
UW-Madison \\
\tt{tzamos@wisc.edu}
}

\date{}

\maketitle

\begin{abstract}

We study black-box reductions from mechanism design to algorithm design for welfare maximization in settings of incomplete information.
Given oracle access to an algorithm for an underlying optimization problem, the goal is to simulate an incentive compatible mechanism.  The mechanism will be evaluated on its expected welfare, relative to the algorithm provided, and its complexity is measured by the time (and queries) needed to simulate the mechanism on any input.
While it is known that black-box reductions are not possible in many prior-free settings, settings with priors appear more promising: there are known reductions for Bayesian incentive compatible (BIC) mechanism design for general classes of welfare maximization problems.
This dichotomy begs the question: which mechanism design problems admit black-box reductions, and which do not?

Our main result is that black-box mechanism design is impossible under two of the simplest settings not captured by known positive results.
First, for the problem of allocating $n$ goods to a single buyer whose valuation is additive and independent across the goods, subject to a downward-closed constraint on feasible allocations, we show that there is no polytime (in $n$) BIC black-box reduction for expected welfare maximization.
Second, for the setting of multiple single-parameter agents---where polytime BIC reductions are known---we show that no polytime reductions exist when the incentive requirement is tightened to Max-In-Distributional-Range.
In each case, we show that achieving a sub-polynomial approximation to the expected welfare requires exponentially many queries, even when the set of feasible allocations is known to be downward-closed.

 \end{abstract}

\section{Introduction}

A central question in algorithmic mechanism design is determining whether a class of incentive compatible mechanism design problems is or is not more computationally difficult than its algorithmic counterpart.  One way to address this question is with a black-box reduction from mechanism design to algorithm design.  Given oracle access to some algorithmic solution to an allocation problem, is it possible to implement an incentive compatible mechanism with similar performance?  Such reductions, when they exist, show not only that mechanism design is no harder than algorithm design, but also that one can handle separately the algorithmic and economic aspects of a given problem.  

A black-box transformation is specified by a number of parameters.  There is an algorithmic optimization problem, where the goal is to maximize some objective subject to a feasibility constraint.  For example, one might want to maximize expected welfare in an allocation problem, where (a) the agents have valuations for goods drawn from some known distribution, (b) the goal is to assign goods in order to maximize aggregate value, and (c) the assignment is subject to some packing constraint(s).  On top of this there is a desired incentive property, such as Bayesian incentive compatibility.  A black-box transformation is given query access to an algorithm $\A$ for the algorithmic optimization problem.  On a given input, the transformation must simulate a mechanism $\Mech$ that satisfies the incentive property, while (approximately) preserving the objective value of $\A$.  Importantly, the transformation has access to the input distribution, but \emph{not the feasibility constraint}.  In this way, the transformation must rely on $\A$ to solve the underlying optimization, and cannot simply disregard $\A$ in favor of implementing an unrelated incentive compatible mechanism.  We measure the runtime of the reduction in terms of the maximum number of queries made to $\A$, plus any side computation, to simulate the mechanism on any single input.

Black-box reductions with priors have been most commonly studied for welfare maximization problems.  Hartline and Lucier~\cite{HL15} initiated this line of study by providing a black-box reduction from Bayesian incentive compatible mechanism design to algorithm design for welfare maximization problems with single-dimensional agent types.  Subsequently, Hartline, Kleinberg, and Malekian~\cite{HKM15} and Bei and Huang~\cite{BH11} provided black-box reductions for $\epsilon$-approximate Bayesian incentive compatibility for arbitrary (i.e., multi-dimensional) agent types, and this was recently improved to a reduction for exact BIC for multi-dimensional types~\cite{DHKN2017}.  The runtime of these reductions depends on a measure of the size of each agent's type space; we will discuss this in more detail later.  On the other hand, it is known that prior-free reductions for DSIC mechanism design cannot approximately preserve worst-case approximation guarantees, even for single-parameter settings~\cite{CIL12}.  Between these known possibility and impossibility results, a rich space of natural open questions still lies open.

\subsection{Our Results}

In this work we present two impossibility results for black-box reductions in settings of incomplete information.
Our main result is an impossibility theorem for multi-dimensional BIC welfare maximization: for the problem of selling multiple goods to additive buyers with independent values across items, any polynomial-time BIC black-box reduction must degrade the expected welfare by a polynomial factor.  This is true even for the special case of a single buyer, in which case BIC is equivalent to DSIC.  Along the way, we show a similar impossibility result for expected welfare maximization with single-parameter agents, when the incentive constraint is tightened from BIC to Max-In-Distributional-Range (MIDR), a strengthening of dominant strategy incentive compatibility~\cite{DD09}.\footnote{An allocation rule is MIDR if the distribution over outcomes returned on input $x$ is at least as good, for $x$, as the distribution returned on any other input.}  Importantly, the welfare target in this latter result is to approximate the expected welfare of the provided algorithm, over the given prior, rather than match its worst-case approximation factor (for which an impossibility result was previously known).

To illustrate the setting of our main result,
consider the following toy example.  A small cloud infrastructure platform serves a single large client, and we are tasked with building its pricing mechanism.
The cloud services are offered as Virtual Machines (VMs), which come in a variety of types (high-memory, CPU-intensive, etc.).  The customer can purchase multiple VMs of possibly multiple different types, and their valuation is additive over VMs, although their value for each different type of VM is private information and drawn independently from some distribution.  However, the VMs themselves are simply abstractions: a VM maps onto a collection of hardward resources, possibly in many different ways.  Whether the platform can feasibly serve a given multiset of VMs is driven by its internal infrastructure technology, and the precise details of this are opaque to us.  An engineering team has provided an algorithm that maps any given customer valuation to a proposed allocation, but this algorithm comes with no incentive guarantees.  Our job is to implement an incentive compatible mechanism for the user, using this algorithm as an oracle and matching its expected welfare.

More generally, we are interested in black-box reductions from IC mechanism design to algorithm design for allocating $n$ goods to a single agent\footnote{Our results in this setting are impossibility results, so they also apply to the case of multiple agents.} with additive valuation and independent values for each good, and subject to some unknown downward-closed feasibility condition on the set of allowable allocations, with the goal of maximizing social welfare.  Our notion of incentive compatibility is dominant strategy incentive compatibility, although we note that since there is only a single agent this corresponds with Bayesian incentive compatibility.  That is, the mechanism offers allocations and payments such that the agent maximizes their expected outcome by reporting their true preferences.

Our main result is that any polytime black-box reduction that is DSIC, either with high probability or in expectation over any randomization in the transformation, 
will necessarily degrade the expected welfare of some algorithms by a factor of $O(n^{1/4})$. Our lower bound makes use of the matching characterization of incentive compatibility due to Hartline, Kleinberg, and Malekian~\cite{HKM15}.  Roughly speaking, any IC transformation must guarantee a certain max-in-range-style property across all subsets of agent types.  However, since the space of types is high-dimensional, it is possible to hide non-monotonicities that cannot be detected by the transformation with a polynomial amount of exploration or sampling.  Since the transformation cannot distinguish scenarios with or without such non-monotonicities, it can achieve incentive compatibility only by aggressively reducing allocations nearly everywhere.  As a corollary, a polytime BIC reduction for the case of $n$ items and $m$ agents must likewise degrade expected welfare by a polynomial factor for some algorithms.

At first glance, this impossibility result may seem at odds with known multi-dimensional reductions for BIC welfare maximization discussed earlier~\cite{HKM15,DHKN2017}.  These known reductions have runtime that depends on the support size of each agent's typespace.  The key distinction in our setting is that the type space for an additive agent with independent values is inherently very large, being exponential in the number of goods (or, for the continuous case, having polynomial doubling number~\cite{HKM15,DHKN2017}).  But the types can be described succinctly, due to the assumed independence across items.  This manifests as an exponential query complexity in the known BIC reductions for multi-dimensional welfare.  Our result can therefore be interpreted as showing that this dependence on the support size is unavoidable, even for a single additive agent.

Along the way, we show an impossibility result for a related setting.  Consider a scenario with $n$ single-parameter agents, each with an independently drawn value for service, and the problem is to determine which subset of agents receive service subject to some downward-closed feasibility constraint.  In this setting, it is known that any DSIC black-box reduction must sometimes degrade the worst-case approximation factor of the allocation oracle by a polynomial factor~\cite{CIL12}.  But it is an open question whether a DSIC black-box reduction could (approximately) preserve the \emph{expected} welfare of a general algorithm with respect to the distribution over types.  We make progress toward this question by showing that any black-box reduction from Maximum-In-Distributional-Range (MIDR) mechanism design to algorithm design that runs in polynomial time must sometimes reduce the expected welfare by a factor of $\Omega(n^{1/4})$.  The MIDR condition, proposed by Dughmi and Dobzinski~\cite{DD09}, implies dominant strategy incentive compatibility; it is strictly stronger than DSIC, but covers a very broad range of techniques for constructing DSIC mechanisms.  As it turns out, the matching characterization of IC, used in our main result for a single multi-dimensional additive agent, is very related to the MIDR condition for $n$ single-parameter agents.  Our result leverages this relationship, and indeed our lower-bound constructions for the two settings are very similar.  
We take this result as suggestive that there is no black-box reduction from DSIC mechanism design to algorithm design for single-parameter agents and the objective of expected welfare.  We conjecture that this is true, but leave it as an open problem.

\subsection{Further Related Work}

As described above, BIC black-box reductions are known for a rich class of welfare maximization problems~\cite{HL15,HKM15,BH11,DHKN2017}.  In the prior-free setting, Babaioff, Lavi, and Pavlov also show that for a class of single-minded combinatorial auction problems, one can achieve DSIC in a black-box way by losing a factor that is logarithmic in the ratio between the largest and smallest possible agent values~\cite{BLP09}.  Dughmi and Roughgarden show a black-box reduction for FPTAS algorithms that also applies in a broad range of multi-dimensional welfare maximization problems~\cite{DR14}.  There is also a significant line of work studying general methods for converting certain types of algorithms into IC mechanisms~\cite{BKV05,LS05}.

For the goal of revenue maximization, Cai, Daskalakis and Weinberg provide a BIC black-box reduction for the case where buyers have valuations that are additive and independent across a collection of goods~\cite{CDW12,CDW13}.  Their approach reduces revenue maximization to welfare maximization by developing an appropriate notion of virtual valuations.  Like the multi-dimensional reductions for welfare described above~\cite{HKM15,BH11,DHKN2017}, the running time of their reduction is polynomial in the number of types of any agent.  Our constructions use of the same setting of additive valuations with independent values across goods.  They also show how to improve the runtime of their reduction to be polynomial in the number of agents and items under an item-symmetry condition~\cite{CDW12}.  Our results imply that such a condition is necessary for polynomial runtime even for welfare maximization.

The first impossibility result for black-box reductions in mechanism design was due to Chawla, Immorlica, and Lucier~\cite{CIL12}, who showed that no black-box reduction that guarantees dominant strategy incentive compatibility (DSIC) can approximately preserve the worst-case approximation factor of a given algorithmic oracle.  Relative to that result, we relax the performance evaluation from worst-case welfare approximation to expected welfare approximation, and strengthen the incentive compatibility constraint from DSIC to MIDR.  In addition to this result, Chawla, Immorlica and Lucier also showed that there is no black-box reduction for BIC mechanisms with the objective of minimizing the makespan for single-parameter agents~\cite{CIL12}.  They left as an open question whether there exists a black-box reduction for DSIC mechanisms with the objective of maximizing expected welfare for single-parameter agents.  We show that the answer is no for the stronger incentive property of MIDR.

We focus primarily on black-box reductions with priors.  For the setting without priors, Pass and Seth~\cite{PS14} build upon the impossibility result of Chawla, Immorlica, and Lucier to show that even if the transformation is given access to the underlying problem's feasibility constraint, black-box transformations are still impossible under standard cryptographic assumptions.  Suksompong~\cite{S18} considers the case of downward-closed single-parameter environments, and likewise shows limits on the power of black-box reductions.

Our negative result for BIC black-box reductions for an additive bidder applies in a setting with a downward-closed feasibility constraint on the allocations.  This is closely related to models of agent valuations that are additive subject to a downward-closed constraint, such as $k$-additive valuations or additivity subject to a matroid constraint.  These models have recently attracted interest in the literature on revenue maximization; for example, it is known that simple pricing methods can achieve a constant fraction of the optimal revenue in any such environment~\cite{RW18}.  Our impossibility result shows that even for the conceptually simpler goal of maximizing expected welfare, no general reduction is possible when the downward-closed constraint is not known to the transformation.

\section{Model}
\label{sec:model}

We begin with some preliminaries before presenting our impossibility results for MIDR reductions for single-parameter agents (Section~\ref{sec:midr}) and for BIC reductions for an additive buyer (Section~\ref{sec:bic}). We conclude with future research directions in Section~\ref{sec:conclusion}.

\subsection{Preliminaries}

A mechanism $\Mech = (\A,P)$, consists of (a) an allocation rule $\A: \Domain^n \rightarrow \Range$ that takes as input the preferences of $n$ agents, each with a type space $\Domain$, and outputs an outcome $y \in \Range$; and (b) a payment function $P: \Domain^n \rightarrow \reals^n$ that computes the payment for every agent.  An agent's type is represented as a valuation function $v_i: \Range \rightarrow \reals$ that assigns a non-negative real value to each outcome in $\Range$.  Given input vector $v \in \Domain^n$, 
each agent $i$ derives value $v_i(A(v))$ from the mechanism's outcome.  Agents are quasi-linear, so that the utility of agent $i$ on input $v$ is $v_i(A(v)) - P_i(v)$.  

The social welfare function $\Wel: \Range \times \Domain^n \rightarrow \reals$ is given by $\Wel(y,v) = \sum_i v_i(y)$.  For simplicity, we sometimes write $\Wel(\A(v),v)$ as $\Wel(\A,v)$, the welfare of allocation rule $\A$ on input $v$.  Given a distribution $\Dist$ over valuations, we write $\Wel(\A) \triangleq \Exp_{v \sim D} [ \Wel(\A(v),v) ]$ for the expected welfare of $\A$.

A feasibility constraint is a subset of possible allocations $\Alloc \subseteq \Range$.  We are interested in the optimization problem of maximizing (expected) welfare subject to a feasibility constraint.  

\subsection{Special Case: Resource Allocation}

All of our lower-bound constructions will use the following class of resource allocation problems.  There are $m$ types of goods and $n$ agents.  
We can think of $\Range$ as a subset of $\reals^{n \times m}$, so that $y \in \Range$ is a vector $(y_{ij})$ and $y_{ij}$ is the amount of good $j$ allocated to agent $i$.
Valuations are additive, with $x_{ij}$ denoting the value of agent $i$ per unit of good $j$.  In other words, $\Domain = \reals^m$, and $v_i(y) = \sum_j x_{ij} y_{ij}$.

Given $y,z \in \Alloc$, we will write $z \leq y$ to mean $z_{ij} \leq y_{ij}$ for all $i$ and $j$.  We say that a feasibility constraint $\Alloc$ is \emph{downward-closed} if, whenever $y \in \Alloc$, we must also have $z \in \Alloc$ for all $z \leq y$.  That is, if a given allocation is feasible, then any other allocation obtained by taking goods away from agents is also feasible.

\subsection{Incentive Constraints}

We call a mechanism \emph{dominant strategy incentive compatible} (DSIC) if for
every agent $i$, any $v \in \Domain^{n}$ and any $v'_i \in \Domain$, it holds
that $v_i(A(v)) - P_i(v) \ge v_i(A(v'_i,v_{-i})) - P_i(v'_i,v_{-i})$.  I.e., no
agent has an incentive to misreport his true preference, for any declaration of
the other agents.

The DSIC constraint requires that every agent prefers to declare his type for
any types reported by the other agents. A weaker notion of incentive
compatibility is Bayesian incentive compatibility, where every agent prefers to
declare his type in expectation over the other agent reports. In particular, if
agent types come according to a product distribution $D$ over $\Domain^n$, we
call a mechanism \emph{Bayesian incentive compatible} (BIC) if for every agent
$i$ and any $v_i,v'_i \in \Domain$, it holds that $\Exp_{v_{-i} \sim
\Dist_{-i}} [ v(A(v_i,v_{-i})) - P_i(v_i,v_{-i}) ] \ge \Exp_{v_{-i} \sim
\Dist_{-i}} [ v_i(A(v'_i,v_{-i})) - P_i(v'_i,v_{-i}) ]$, where $\Dist_{-i}$ is
the distribution of all agent types other than $i$.  Notably, in the special
case where there is only a single agent, the definitions of DSIC and BIC
coincide.

We will also consider the class of ``maximal-in-distributional-range''
(MIDR) mechanisms~\cite{DD09}.

\begin{definition}
A randomized algorithm $A$ is maximal-in-distributional-range (MIDR) if, for
all inputs $v,v' \in \Domain$, it holds that $\Exp_{y \sim \A(v)}[ \Wel(y,v) ] \ge
\Exp_{y \sim \A(v')}[ \Wel(y,v) ]$.
\end{definition}

In other words, an allocation rule is MIDR if, for every input profile $v$, the expected 
welfare obtained on input $v$ would not be increased by instead returning the (random) allocation 
on any other input.  The MIDR property is a strengthening of DSIC, so any MIDR mechanism is
DSIC when paired with appropriate payments~\cite{DD09}.

We will also be interested in DSIC mechanisms with a single agent.
Hartline, Kleinberg, and Malekian~\cite{HKM15} provide a characterization of all the 
allocation rules that can be converted to BIC mechanisms.  Since BIC and DSIC coincide for a single agent, this
characterization also applies to DSIC mechanisms with a single agent.  This is as follows:

for a subset of types $\Domain' \subseteq \Domain$, consider the following weighted
bipartite graph $G(\Domain')$ constructed as follows.  The vertices of $G(\Domain')$ are 
$(U,V)$, where $U = V = \Domain'$.  That is, each side of the bipartite graph is a copy of $\Domain'$.  
For $v,w \in \Domain'$, the weight of
edge $(v,w)$ is taken to be $\Exp_{z \sim A(w)}[ v \cdot z ]$.  That is, the
weight of $(v,w)$ is the value that type $v$ has for the (possibly randomized)
outcome generated for type $w$.

A matching $m$ of $G(\Domain')$ is subset of edges of $G(\Domain')$ such that each vertex is 
incident with at most one edge.  For a given vertex $v$ incident with an edge in the matching, 
we will write $m(v)$ for the node matched with $v$.

\begin{theorem}[\cite{HKM15}]\label{thm:match}
Mechanism $A$ is BIC if and only if, for every $\Domain' \subseteq \Domain$,
the maximum weighted matching in $G(\Domain')$ is the identity matching, $m(v)
= v$ for all $v \in \Domain'$.
\end{theorem}

\subsection{Black-Box Transformations}
We consider the following general setup for black-box transformations.  We are
given sample access to a prior distribution $\Dist$ over $\Domain^n$, and the
designer's objective is to maximize expected welfare subject to a feasibility
constraint $\Alloc \subseteq \Range$.  The feasibility constraint $\Alloc$
is not known.
Instead, we are given black-box access to a possibly randomized algorithm $\A:
\Domain^n \rightarrow \Range$ that returns allocations in $\Alloc$.  We can
query $\A$ on any input and observe the resulting allocation.  The designer's
goal is to implement a mechanism $\Mech_{\A} = (\A',P)$ with the following
guarantees:
\begin{itemize}
  \item $\A'$ only returns outcomes in $\Alloc$.
  \item $\Mech_{\A}$ satisfies a desired incentive guarantee, e.g., BIC or DSIC.
  \item $\A'$ achieves comparable welfare to $\A$.   
  That is, $\Wel(A')$ should approximate $\Wel(A)$.
\end{itemize}

We will tend to write $\Mech$ for a black-box transformation, which is defined with respect to oracle access to an allocation algorithm.  We will write $\Mech_{\A}$ for the mechanism implemented by this transformation given oracle access to algorithm $\A$.

Let us briefly comment on the first requirement, feasibility.  Since $\Alloc$ is unknown, 
the requirement that range$(\A) \subseteq \Alloc$ can only be satisfied by using algorithm $\A$ to learn feasible outcomes. 
One way to guarantee feasibility 
is to only ever return outcomes that are output by $\A$ on some input.  For general feasibility
constraints, this is the \emph{only} way to guarantee feasibility.  However, in some cases,
the designer may know that $\Alloc$ satisfies some property (such as downward-closedness), in which
case it's possible to determine an outcome is feasible without observing it as an
output of $\A$.  For example, if $\A$ is observed to allocate to agents $1$ and $2$ simultaneously,
then for a downward-closed constraint one can infer that it is also feasible to allocate only to agent $1$.

\subsubsection{Query Complexity}

We emphasize that our interest is in \emph{computationally efficient} black-box transformations that guarantee
our desired incentive and objective constraints.  For each of the incentive constraints we consider, it is straightforward to
construct a transformation without any loss of welfare if there were no contraints
on runtime or query complexity.
For example, one can transform any algorithm into a MIDR mechanism by first
querying the algorithm on every possible input, then returning the welfare-maximizing
allocation from the resulting range of outcomes.
However, this transformation is impractical when the number of possible input profiles
is large.
We study whether more efficient transformations exist that only need to query $\A$
at a polynomial (in $n$ and $m$) number of points in order to simulate an IC mechanism
on any given input. 

To this end, we define the query complexity of a transformation $\Mech$ to be the maximum number of
queries to the original function $\A$ that need to be performed in order to
compute the value of the transformed function $\A'(x)$ at any given point $x$.

\begin{remark}
  While the definition of the problem allows the input algorithm $\A$ to be randomized,
  it is without loss of generality to assume that it is deterministic. Indeed, if we consider the
  deterministic algorithm $\tilde A$ obtained by sampling $\A$ on every input, the expected welfare of
  $\tilde A$ will be indentical to that of $\A$.
  As $\tilde A$ must return the same value at a point $x$ when queried multiple times,
  we can fix the random seed of the algorithm $\A$ in advance before performing the conversion.
\end{remark}

\section{A Lower Bound for MIDR Transformations}
\label{sec:midr}

In this section we restrict our attention to a welfare maximization problem
with $n$ single-dimensional agents, meaning that there is only a single type of
good (i.e., $m = 1$).  The optimization is subject to a downward-closed
feasibility constraint.  Indeed, it will actually suffice to consider binary
types and binary outcomes: the type space will be $\Domain = \{0,1\}^n$, and
the outcome space is $\Range = \{0,1\}^n$.  That is, each agent either has a
unit of value for being served, or no value for being served; and each agent
either receives a unit of service or none.  

We can think of an input profile as a subset of agents, corresponding to those with value $1$.  Similarly, an allocation is also associated with a subset of agents: those who receive a unit of service.  With this in mind, we will sometimes represent inputs and outcomes as sets for notational convenience.

We will show that even in this very restricted setting, any black-box reduction from single-dimensional MIDR mechanism design to algorithm design with polynomial query complexity must degrade the expected welfare by a polynomial factor.

\begin{theorem}
\label{thm:main.midr}
  There are constants $\epsilon > 0$ and $c > 0$ such that the following is true.
  For any MIDR black-box transformation $\Mech$ with query complexity $e^{n^\epsilon}$, 
  there exists an algorithm $A$ and distribution $\Dist$ 
  such that $\Wel(\Mech_A) \le \frac{ \Wel(A) }{ n^{c} }$.  This is true even if the feasibility constraint $\Alloc$ is known to be downward-closed.
\end{theorem}

Before going into the full details of the lower bound construction, let us describe the high-level intuition behind Theorem~\ref{thm:main.midr}.  We will have $n$ agents with binary types, where each agent has value $1$ with probability roughly $1/\sqrt{n}$, so that the expected number of agents with value $1$ is roughly $\sqrt{n}$.  Outcomes are also binary, with each agent either getting service or not.  The allocation algorithms we consider are very simple: they just serve all agents with value $1$ (as long as there aren't too many), subject to the following exception.  There are two hidden subsets of agents $S$ and $T$, each of size roughly $\sqrt{n}$, which have a non-trivial intersection.  If the set of agents with value $1$ contains many agents from $T$, but too few agents from $S$ (where $|S \cap T|$ counts as ``too few''), then the algorithm will instead return the zero allocation.  We will write $\A_{S,T}$ for the algorithm with a certain choice of the hidden sets $S$ and $T$.  See Figure~\ref{fig:alg} for an illustration of $\A_{S,T}$.  What we will show is that for any polytime MIDR black-box transformation, there is some choice of $S$ and $T$ for which $\Mech$ substantially reduces the welfare of $\A_{S,T}$.

We establish the result in three steps.  First, we argue that when given black-box access to a random $\A_{S,T}$, a polytime MIDR transformation will generate low expected welfare on input $T$ with high probability.  This step is formalized in Lemma~\ref{lem:mech.T} below.  To see why this is true, consider what the transformation can do on input $T$ (i.e., when the agents in $T$ are precisely the ones with value $1$).  On this input, the oracle $\A_{S,T}$ returns the zero allocation.  And since $S$ and $T$ are chosen at random, the transformation does not know the identity of set $S$ on input $T$.  Moreover, because of the high dimensionality of the space, the transformation will not be able to find set $S$ with any polynomial number of samples.  This means that any input it queries with large intersection with $T$ will generate the zero allocation with high probability.  Because of this, the transformation will necessarily generate low expected welfare on input $T$.  Note that this does not directly imply that the mechanism has low welfare overall, since we have so far argued only about its output on a single input.

In the second step, we argue that on any input with a large enough intersection with $T$, the mechanism must again generate low welfare.  This step is given as Lemma~\ref{lem:mech.S}.  This follows from the MIDR property: since the mechanism generates low welfare on input $T$, the distribution over outputs generated on \emph{any} input must also have low welfare for input $T$.  On input $S$, for example, the oracle $\A_{S,T}$ allocates to all agents in $S$, but this would generate welfare $|S \cap T|$ for input $T$, which (if we choose our parameters carefully) will be too high and violate the MIDR condition.  So the transformation must avoid allocating too much to the agents in $S \cap T$.  This does not complete the proof of Theorem~\ref{thm:main.midr}, though, since most inputs do not have high intersection with $T$.

In the third and final step, we argue that the mechanism must generate low welfare on almost all inputs.  We establish this step in Lemma~\ref{lem:mech.welfare}.  The idea is that, by the curse of dimensionality, the transformation cannot determine by sampling whether or not a random input has high intersection with $T$.  In particular, the only way to determine whether an input vector contains $S \cap T$ is to find an input very close to $T$ (on which $\A_{S,T}$ returns $\emptyset$), which requires a super-polynomial number of samples.  This means that, on most inputs, the transformation cannot rule out that the input contains $S \cap T$.  However, by the previous step, the mechanism must generate low welfare on any such input.  This ultimately leads to the catastrophic loss of welfare in Theorem~\ref{thm:main.midr}.

\medskip

We now proceed with the details of the proof of Theorem~\ref{thm:main.midr}.
Choose some sufficiently small constant $\epsilon > 0$.  For notational convenience, we will define constants $N$, $\eps_{ST}$, $\eps_S$, and $\eps_T$.  These constants will be used to formalize the allocation rule $\A_{S,T}$ described informally above.  Roughly speaking: $N$ is the maximum number of agents who can declare $1$ before the algorithm returns the zero allocation.  Constant $\eps_{ST}$ will define the size of the intersection $|S \cap T|$.  Constants $\eps_S$ and $\eps_T$ determine what is meant by ``too few agents from $S$'' and ``too many agents from $T$,'' respectively, in the description of $\A_{S,T}$ given above.
We will set these to be $\eps_{ST} \triangleq n^{-1/4}$, $\eps_S \triangleq 2n^{-1/4}$, $\eps_{T} \triangleq 16n^{-1/2}$, and $N \triangleq n^{1/2 + 2\epsilon}$.

The distribution $\mathcal{D}$ selects $x\in \Domain$ by independently drawing every coordinate $x_i \sim \text{Bernoulli}\left( \frac {3N} {4n} \right)$.

We say that a pair of subsets of coordinates $S, T \subseteq [n]$ is \emph{valid} if $|S| = |T| = \frac{N}{2}$ and $|S \cap T| = \eps_{ST} N$.
We consider the following family of algorithms, parameterized by a valid pair of subsets $S$ and $T$.

$$A_{S,T}(x) = \begin{cases}
   x &\text{if } |x| \le N \text{ and } \bigg(|x \cap T| \le \eps_T N \text{ or } |x \cap S| \ge \eps_S N \bigg)  \\
   \emptyset &\text{otherwise}
\end{cases}$$

\begin{figure}[t]
  \centering
  \begin{tikzpicture}[scale=0.80]
\tikzset{ST/.style={green, dashed, very thick}}
\tikzset{myset/.style={red, very thick}}
\tikzset{square/.style={black, very thick}}

\tikzset{rx/.style={x radius=#1},ry/.style={y radius=#1}}

\pgfmathsetmacro{\length}{6.5}
\pgfmathsetmacro{\height}{4}
\pgfmathsetmacro{\rad}{(\length-0.5)/4}
\pgfmathsetmacro{\rx}{2}
\pgfmathsetmacro{\ry}{5}

\def\lab{{"a","b","c", "d"}}

\def\xcoords{{0, \length+1,		    0,  \length+1}}
\def\ycoords{{0, 		 0,-\height-1, -\height-1}}

\foreach \i in {1,2,3,4} 
{ 

	\pgfmathsetmacro{\labelz}{\lab[\i-1]}
	\pgfmathsetmacro{\x}{\xcoords[\i-1]}
	\pgfmathsetmacro{\y}{\ycoords[\i-1]}

	\draw[square] (\x, \y) rectangle (\x + \length,\y + \height) node[] at (\x + \length/2,\y - 0.5){(\labelz)};
	
	\draw[ST] (\x + \rad + 1,\y + 1) circle [rx=\rad cm, ry=0.5cm] node[black, left=1.2cm]{$S$};
	\draw[ST] (\x + \length - \rad -1,\y + 1) circle [rx=\rad cm, ry=0.5cm] node[black, right=1.2cm]{$T$};

}

\pgfmathsetmacro{\x}{\xcoords[0]}
\pgfmathsetmacro{\y}{\ycoords[0]}
\draw[myset] (\x + \length/2, \y+\height-\height/3) circle[rx = \rad cm, ry = 0.5cm] node[black]{$|x|> N$};
\node[red] at (\x+0.9*\length, \y + 0.9*\height) {\Huge \xmark};

\pgfmathsetmacro{\x}{\xcoords[1]}
\pgfmathsetmacro{\y}{\ycoords[1]}
\draw[myset] (\x + \length/2, \y+\height-\height/3) circle[rx = \rad cm, ry = 0.5cm] node[black]{$|x|\leq N$};
\node[green] at (\x+0.9*\length, \y + 0.9*\height) {\Huge \checkmark};

\pgfmathsetmacro{\x}{\xcoords[2]}
\pgfmathsetmacro{\y}{\ycoords[2]}
\draw[myset] (\x + 2*\length/3, \y+2) circle[rx = 0.6cm, ry = \rad cm];
\node[black] at (\x+2, \y+\height-1) {$|x\cap T| > \epsilon_T N$};
\node[black] at (\x+2, \y+\height-1.5) {$|x\cap S| < \epsilon_S N$};
\node[red] at (\x+0.9*\length, \y + 0.9*\height) {\Huge \xmark};

\pgfmathsetmacro{\x}{\xcoords[3]}
\pgfmathsetmacro{\y}{\ycoords[3]}
\draw[myset] (\x + \length/2, \y + 1.5) circle[rx = 0.9cm, ry = \rad cm, rotate=-60];
\node[black] at (\x+2, \y+\height-1) {$|x\cap T| > \epsilon_T N$};
\node[black] at (\x+2, \y+\height-1.5) {$|x\cap S| \geq  \epsilon_S N$};
\node[green] at (\x+0.9*\length, \y + 0.9*\height) {\Huge \checkmark};

\end{tikzpicture}
   \caption{Allocation rule $\A_{S,T}$, parameterized by sets $S$ and $T$ (dashed green).  On input $x$ (solid red), the allocation is either the all-zero allocation (red X) or the set $x$ itself (green checkmark).  (a) Any $x$ with $|x| > N$ returns an empty allocation.  (b) Most sets $x$ with $|x| \le N$ return allocation $x$, but (c) if $x$ has a large intersection with $T$ then $\A_{S,T}$ returns the empty allocation, unless (d) $x$ also has a large intersection with $S$.}
  \caption{An illustration of the allocation rule used to prove Theorem~\ref{thm:main.midr}.}
  \label{fig:alg}
\end{figure}

In other words, $A_{S,T}(x)$ simply allocates $1$ to every agent who declared value $1$, as long as there are at most $N$, unless those agents overlap significantly with $T$ but not with $S$.  See Figure~\ref{fig:alg} for an illustration of the allocation rule $A_{S,T}(x)$.
We will define the feasible allocations associated with $S$ and $T$, $\Alloc_{S,T}$, to be all allocations in the range of $A_{S,T}$ and their downward closure.  That is, $R \in \Alloc_{S,T}$ if and only if there exists some $R' \in \text{range}(A_{S,T})$ such that $R \subseteq R'$.

We will make repeated use of the following lemma, which follows directly from the Chernoff-Hoeffding bound.
\begin{lemma}\label{lem:chernoff}
Suppose $x_1, \dotsc, x_k$ are i.i.d.\ binary random variables, with $X = \sum_i x_i$.  Then for any constant $Y \geq 3 \cdot E[X]$, we have $\Prob[ X > Y ] < e^{-Y/4}$.   
\end{lemma}
\begin{proof}
Write $Y = (1+\delta)E[X]$.  Then by the Chernoff-Hoeffding bound, $\Prob[ X > (1+\delta)E[X] ] < e^{ - \delta^2 E[X] / (2+\delta) }$.  The result follows by noting that $(\delta+1)/4 \leq \delta^2 / (2+\delta)$ for all $\delta \geq 2$.
\end{proof}

To prove Theorem~\ref{thm:main.midr} we will show that, for any MIDR transformation $\Mech$ with query complexity $e^{n^\epsilon}$, there is an algorithm $A_{S,T}$ satisfying the conditions of the theorem.  
We begin by noting that, for any valid pair $S$ and $T$, algorithm $A_{S,T}$ has welfare at least $N/4$. 

\begin{lemma}\label{lem:alg.welfare}
  For any valid pair $S$ and $T$, $\Wel(A_{S,T}) \ge  \frac N {4} = \Omega(N)$.
\end{lemma}
\begin{proof}
  Note first that Chernoff-Hoeffding bounds imply that 
 \[
	\Prob_{x \sim \Dist}[ |x| \in [N/2, N]] = \Prob_{x \sim \Dist}[ |x| - \E[|x|] > N/4 ] < e^{- N / 64} < 1/4
 \] 
  for sufficiently large $N$.  

  Next, note that each of the $N/2$ coordinates in $T$ will lie in $x$
independently with probability $\tfrac{3N}{4n}$.  The expected size of $x \cap
T$ is therefore $\tfrac{3N^2}{8n}$.  By Lemma~\ref{lem:chernoff}, $\Prob_{x
\sim \Dist}[ |x \cap T| > \eps_T N] < e^{-\eps_T N / 4}$ since $\eps_T N >
\tfrac{9N^2}{8n}$.  Again, this probability is at most $1/4$ for sufficiently
large $N$.

  By taking a union bound over these events, we conclude that $|x| \in [N/2, N]$ and $|x \cap T| \leq \eps_T N$ with probability at least $1/2$.  This implies that $A_{S,T}(x) = x$, with $|x| \geq N/2$, with probability at least $1/2$.  The expected welfare of $A_{S,T}$ is therefore at least $N/4$.
\end{proof}

Fix a MIDR black-box transformation $\Mech$.  For notational convenience, we will write $\Mech_{S,T} = \Mech_{A_{S,T}}$.  To prove Theorem~\ref{thm:main.midr}, we wish to find a valid pair $S$ and $T$ such that $\Wel(\Mech_{S,T})$ is small.  To this end, we will imagine selecting the sets $S$ and $T$ uniformly at random, subject to being a valid pair.  Write $\Gamma$ for the uniform distribution over all valid pairs $(S,T)$.  Note that drawing $(S,T) \sim \Gamma$ is equivalent to the following process: first choose the $N/2$ coordinates of $S$, then choose the $\eps_{ST} N$ coordinates of $S$ to form $T \cap S$, then choose the $(1/2 - \eps_{ST})N$ coordinates of $[n] \backslash S$ to form $T \backslash S$.  We will write $\Gamma_T$ for the uniform distribution over sets $S$ such that $(S,T)$ is a valid pair, and similarly for $\Gamma_S$.

The next lemma shows that, with high probability (over the random choice of $S$ and $T$ and over any randomness in the transformation), $\Mech_{S,T}$ will have low welfare \emph{on input $T$.}  Note that this lemma does not make use of any incentive properties of $\Mech_{S,T}$.  Rather, it follows from the assumption that $\Mech_{S,T}$ has query complexity at most $e^{n^\epsilon}$, which is not enough samples to find a high-welfare outcome on input $T$.

\begin{lemma}\label{lem:mech.T}
  $\Prob_{(S,T)\sim \Gamma}[\Wel(\Mech_{S,T}, T) > \eps_T N] < e^{n^\epsilon} \cdot e^{-(\epsilon_S - \epsilon_{ST})N/4} = O(e^{-n^{1/4}})$.
\end{lemma}

\begin{proof}
Choose some valid pair $S$ and $T$.  In order for $\Wel(\Mech_{S,T}, T)$ to be
greater than $\eps_T N$, it must be that on input $T$, the transformation must
query at least one input $x$ such that $|A_{S,T}(x) \cap T| > \eps_T N$.  
Any such $x$ must satisfy $|x| \leq N$, $| x \cap T| > \eps_T N$, and $A_{S,T}(x)
= x$.  From the definition of $A_{S,T}(x)$, this can occur only if $|x \cap S|
\geq \eps_S N$.
And since $|S \cap T| = \eps_{ST} N$, this requires that $|x \cap (S \setminus T)| \ge (\eps_S - \eps_{ST}) N$.

Now fix any $x$ and any set $T$ such that $|T| = N/2$ and $|x \backslash T| \leq |x| \leq N$, and consider the event $|x \cap (S \setminus T)| \ge (\eps_S - \eps_{ST}) N$ with respect to a randomly chosen $S \sim \Gamma_T$.  As each coordinate in the complement of $T$ is in $S \setminus T$ independently with probability $p \triangleq \frac{|S \setminus T|}{n - |T|} = \frac{(1/2 - \epsilon_{ST})N}{n - N/2}$, the event that $|x \cap (S \setminus T)| \ge (\eps_S - \eps_{ST}) N$ is at most the probability that $N$ indicator variables, each $1$ with probability $p$, sum to at least $(\eps_S - \eps_{ST}) N$.  

By Lemma~\ref{lem:chernoff}, this probability is at most $e^{-(\epsilon_S - \epsilon_{ST})N/4}$, as long as
$$(\epsilon_S - \epsilon_{ST})N > 2N \cdot \frac{(1/2 - \epsilon_{ST})N}{n - N/2}.$$
We note that since $\epsilon_{S} = 2 \epsilon_{ST}$, this is implied by $\epsilon_{ST} > 2N/n$, which is satisfied by the setting of our parameters.

We conclude that, for any $T$ and any query point $x$, $\Prob_{S \sim \Gamma_T}[|A_{S,T}(x) \cap T| \geq \eps_T N] < e^{-(\eps_S - \eps_{ST})N/4}$.  
Taking a union bound over the (at most) $e^{n^\epsilon}$ samples taken by the mechanism yields the desired result.
\end{proof}

We have now established that mechanism $\Mech_{S,T}$ has low welfare on input $T$, with high probability.  The MIDR property therefore implies that $\Mech_{S,T}$ cannot generate outcomes that have high (expected) intersection with $T$, on any input, again with high probability.  Our next lemma uses this to show that,   
with high probability, $\Mech_{S,T}$ will have low welfare on any input $x$ that contains $S$ and for which $x \cap T = S \cap T$.  For convenience, define $\zeta \triangleq 2 e^{n^\epsilon} e^{-\epsilon_T N/6}$.  For our choice of parameters, $\zeta \to 0$ as $n \to \infty$ faster than any polynomial in $n$.

\begin{lemma}\label{lem:mech.S}
Choose $S$ with $|S| = N/2$ and some $x \supset S$ with $|x| \in [N/2, N]$.  Then 
$$\Exp_{T \sim \Gamma_S }[ \Wel(\Mech_{S,T},x) \big\rvert x \cap T = S \cap T] \leq \frac {4N(\eps_T + \zeta)} {\eps_{ST}} = O\left( \frac{N}{n^{-1/4}} \right) .$$
\end{lemma}

\begin{proof}
Fix some $x$ and some $S \subseteq x$ with $|S| = N/2$.  For any $T$ such that $S$ and $T$ are a valid pair, and for which $T \cap S = T \cap x$, the MIDR property applied to mechanism $\Mech_{S,T}$ implies that
  $$
  \Wel(\Mech_{S,T}, T) \ge \Exp_{y \sim \Mech_{S,T}(x)}[ T \cdot y ].
  $$
Taking an expectation over a uniformly random $T$ satisfying the conditions above, we have that:
\begin{equation}
\label{eq.T}
  \Exp_{T\sim \Gamma_S }[ \Wel(\Mech_{S,T}, T) ] 
		\ge \Exp_{T\sim \Gamma_S }[ \Exp_{y \sim \Mech_{S,T}(x)}[ T \cdot y ] ].
\end{equation}
The left hand side is at most $2 \eps_T N$, by Lemma~\ref{lem:mech.T} plus the fact that $\Wel(\Mech_{S,T}, T) \leq |T| = N/2$ unconditionally.  To bound the right-hand side, we will argue that the transformation is unlikely to gain information about the identity of set $T$ while querying on input $x$, over randomness in the transformation and over the choice of set $T$.

To establish this claim, consider a baseline algorithm $A$ such that $A(x) = x$ for all $|x| \leq N$ and $A(x) = \emptyset$ for all $|x| > N$.  That is, $A$ behaves like $A_{S,T}$ except that there are no hidden sets $S$ and $T$.  We will show that the transformation cannot distinguish $A_{S,T}$ from $A$, which implies that the transformation gains no information about the identity of set $T$.  More precisely, we will bound the probability, over $T$ and any randomness in the transformation, that on input $x$, the transformation queries any point $z$ such that $A(z) \neq A_{S,T}(z)$.  In order for $A(z) \neq A_{S,T}(z)$, 
it must be that either
\begin{enumerate}
  \item $|(z \cap T) \setminus S| > \eps_T N / 2$, or 
  \item $|(z \cap T) \cap S| > \eps_T N / 2$ and $|z \cap S| \leq \eps_S N$.
\end{enumerate}
To bound the probability of these events, let us suppose that the transformation knows the identity of set $S$ in addition to the input $x$.  The likelihood of the first event is then maximized if all coordinates of $z$ lie outside of $S$.  Since $T \setminus S$ is uniformly random over the coordinates outside $S$, the probability of the first event is bounded by the probability that $N$ indicator variables, each set to $1$ with probability $\frac{|T \setminus S|}{n - |S|} = \frac{N/2 - \eps_{ST}N}{n - N/2}$, has sum greater than $\eps_T N / 2$.  By Lemma~\ref{lem:chernoff}, this occurs with probability at most $e^{-\epsilon_T N / 6}$, as long as $\epsilon_T N / 2 > 2N \frac{N/2 - \eps_{ST}N}{n - N/2}$.  Note that this is implied by $\epsilon_T > 2N/n$, which is true for our choice of parameters.

The second event is most likely if $|z \cap S| = \eps_S N$, in which case it occurs if the $\eps_S N$ coordinates of $S$ chosen to be in $z \cap S$ contain $\eps_T N /2$ coordinates of $T$.  As each coordinate of $S$ is equally likely to be in $T$, this is precisely the probability that $\eps_S N$ indicator variables, each set to $1$ with probability $|S \cap T| / |S| = 2\eps_{ST}$, has sum greater than $\eps_T N / 2$.  By Lemma~\ref{lem:chernoff}, this probability is at most $e^{-\epsilon_T N/6}$, as long as $\epsilon_T/2 > 4 \epsilon_{ST} \epsilon_S$.  Assuming $\epsilon_{ST} = 2\epsilon_S$, this is satisfied as long as $\epsilon_{T} > 16\epsilon_S^2$, which is true for our choice of parameters.

Taking a union bound over these two events applied to the $e^{n^\epsilon}$
queries made by the mechanism on input $x$, we have that with probability at
least $\zeta \triangleq 2 e^{n^\epsilon} e^{-\epsilon_T N/6}$, $\Mech_{S,T}(x)$
is equal (as a distribution) to $\Mech_A(x)$.  Note that $\zeta$ approaches $0$
with $n$ faster than any polynomial, as long as $\epsilon_T N > n^{2\epsilon}$.
We therefore have that 
\begin{align*}
\Exp_{T\sim \Gamma_S }[ \Exp_{y \sim \Mech_{S,T}(x)}[ T \cdot y ] ] 
& \geq \Exp_{T\sim \Gamma_S}[ \Exp_{y \sim \Mech_A(x)}[ T \cdot y ] ] - \zeta N\\
& \geq \Exp_{T\sim \Gamma_S }[ \Exp_{y \sim \Mech_A(x)}[ T \cdot y \cdot x ] ] - \zeta N\\
& \geq \frac{\epsilon_{ST}}{2} \Exp_{T\sim \Gamma_S}[\Wel(\Mech_A, x)] - \zeta N \\
& \geq \frac{\epsilon_{ST}}{2} \Exp_{T \sim \Gamma_S}[\Wel(\Mech_{S,T}, x)] - 2\zeta N.
\end{align*} 
Where in the first and last inequalities we used that, conditional on being
non-identical, $\Mech_{S,T}$ and $\Mech_A$ can have welfares differing by at
most $N$.  We also used that, from our choice of randomness over $T$, each
coordinate of $x$ is contained in $T$ with probability at most $|S\cap T|/|S| =
\eps_{ST}/2$.  

Applying this inequality to \eqref{eq.T}, we conclude that
\begin{equation}
\label{eq.zeta}
  2 \eps_T N + 2 \zeta N \geq \frac{\epsilon_{ST}}{2} \Exp_{T \sim \Gamma_S}[ \Wel(\Mech_{S,T}, x) ]
\end{equation}
which gives the desired result. 
\end{proof}

Finally, we show that the welfare bound from the previous lemma extends to \emph{any} input that does not have a large intersection with $T$.  This is because the transformation wil not be able to distinguish any such input from one satisfying the conditions of Lemma~\ref{lem:mech.S}, with high probability.

\begin{lemma}\label{lem:mech.welfare}
Fix any $Z$ with $|Z| \in [N/2, N]$.  Then 
$$\Exp_{(S,T) \sim \Gamma }[\Wel(\Mech_{S,T}, Z) \big\rvert |Z \cap T| \leq \epsilon_T N /2 ] 
	\le \frac {5 \eps_T N} {\eps_{ST}} 
	= O\left( \frac{N}{n^{1/4}} \right).$$
\end{lemma}
\begin{proof}
We will apply an argument very similar to the proof of Lemma~\ref{lem:mech.S}.  As in the previous lemma, the probability that the transformation can distinguish between $\Mech_{S,T}$ and $\Mech_A$ on input $Z$ is at most $\zeta \triangleq 2 e^{n^\epsilon} e^{-\epsilon_T N/6}$, even if the transformation is made aware of the identity of set $S$ and of $Z \cap T$, and the probability is only with respect to randomness in $T \backslash Z$.  
We therefore have

\begin{align*}
\Exp_{(S,T) \sim \Gamma}[ \Wel(\Mech_{S,T}, Z) ] 
& \leq \zeta N + \Wel(\Mech_A, Z)\\
& = \zeta N + \Exp_{(S,T) \sim \Gamma}[ \Wel(\Mech_A, Z) \big \rvert S \subseteq Z, S \cap T = z \cap T]\\
& \leq 2\zeta N + \Exp_{(S,T) \sim \Gamma}[ \Wel(\Mech_{S,T}, Z) \big\rvert S \subseteq Z, S \cap T = z \cap T ]\\
& \leq 4N \frac{\eps_T}{\eps_{ST}} + 2\zeta N + 4\zeta N / \eps_{ST}
\end{align*} 
which is at most $5N \frac{\eps_T}{\eps_{ST}}$ for sufficiently large $N$, where in the second-to-last inequality we used the fact that $\Mech_A(Z)$ and $\Mech_{S,T}(Z)$ are identical with probability at least $1-\zeta$ under the conditions of Lemma~\ref{lem:mech.S} (as shown in the proof of that lemma), and the last inequality is Lemma~\ref{lem:mech.S}.
\end{proof}

To complete the proof of Theorem~\ref{thm:main.midr}, we recall that $x \sim \Dist$ satisfies $|x| \in [N/2, N]$ with probability at least $1 - e^{-N/64}$, and the probability that $|Z \cap T| > \epsilon_T N / 2$ is at most $e^{-\epsilon_T N / 6}$.  Thus, by Lemma~\ref{lem:mech.welfare} plus the fact that no allocation can generate welfare greater than $N$, we have that 
$$\Exp_{S,T \sim \Gamma}[\Wel(\Mech_{S,T})] \le 5 \frac{\eps_T}{\eps_{ST}} N + N \cdot (e^{-N/64} + e^{-\epsilon_T N / 6}) \leq 6 \cdot 16 \cdot N / n^{1/4}$$
for our setting of $\eps_T$ and $\eps_{ST}$.  There must therefore be a particular choice of $S$ and $T$ such that $\Wel(\Mech_{S,T}) \leq 6 \cdot 16 \cdot N/n^{1/4}$.  Since $\Wel(A_{S,T}) \geq N/4$ by Lemma~\ref{lem:alg.welfare}, the theorem holds for this particular choice of $A_{S,T}$, for any $c > 1/4$.

\section{A Lower Bound for Multi-Dimensional Transformations}\label{sec:bic}

We now move on to our main result, which is a lower bound for multi-dimensional DSIC transformations.  Our construction uses only a single agent, and therefore also shows a lower bound for multi-dimensional BIC transformations.

To show a lower bound for multi-dimensional transformations, we consider a slightly different setting with a single agent and multiple goods.
There is now a single agent, who has an additive valuation over $n$ goods.\footnote{In a slight abuse of notation, we are using $n$ for the number of goods in this setting rather than the number of agents.  This is intentional, to draw a parallel between this parameter and the number of agents in Section~\ref{sec:midr}.}  A type is still an $n$-dimensional vector of values, say $(x_1, \dotsc, x_n)$, where $x_i$ is the agent's value for good $i$.  The outcome space is binary: an outcome is a binary vector $(y_1, \dotsc, y_n) \in \{0,1\}^n$.  However, unlike in the previous section, agent values are not necessarily binary.  We are still concerned with welfare maximization, and we will focus on mechanisms that are DSIC.  That is, mechanisms for which the agent maximizes their expected utility by declaring their true type. 

\subsection{Construction}

We will set parameters $N$, $\eps_T, \eps_S, \eps_{ST}$ precisely as in Section~\ref{sec:midr} for the proof of Theorem~\ref{thm:main.midr}.
The type space will be $\Domain = \{0,1,\alpha\}^n$, where we will choose $\alpha = 2 \epsilon_{ST}^{-1} > 1$.

The distribution $\mathcal{D}$ selects $x \in \Domain$ by first drawing every coordinate $x_i \sim \text{Bernoulli} \left( \frac{3N}{4n} \right)$.  Then, for each coordinate with $x_i = 1$, we will (independently) instead set $x_i = \alpha$ with probability $p = 1/\alpha$.  Note then that conditional on $x_i$ being non-zero, $E[x_i] < 2$.

Even though types are no longer binary vectors, we can still think of a set associated with each $x \in \Domain$, corresponding to the non-zero indices of $x$.  So, for example, $|x \cap T|$ is the number of indices in $T$ that are non-zero in $x$.
With this understanding, we note that algorithm $A_{S,T}$ from Section~\ref{sec:midr} is well-defined also in our augmented domain, for a given valid pair of sets $S$ and $T$.  
We also note that since the probability that any given coordinate is non-zero matches that from our construction in Section~\ref{sec:midr}, the distribution over sets implied by this type distribution remains unchanged.

\subsection{Adapting the MIDR lower bound}

We will prove the following variation of Theorem~\ref{thm:main.midr}.

\begin{theorem}
\label{thm:main.bic}
  There are constants $\epsilon > 0$ and $c > 0$ such that the following is true for the setting described above.
  For any DSIC black-box transformation $\Mech$ with query complexity $e^{n^\epsilon}$,
   there exists an algorithm $A$ and distribution $\Dist$ 
  such that $\Wel(\Mech_A) \le \frac{ \Wel(A) }{ n^{c} }$.  This is true even if the feasibility constraint $\Alloc$ is restricted to be downward-closed.
\end{theorem}

Our proof closely follows the proof of Theorem~\ref{thm:main.midr}.
The following lemmas are similar to Lemma~\ref{lem:alg.welfare} and Lemma~\ref{lem:mech.T}  from Section~\ref{sec:midr}, and we only sketch the differences in their (nearly identical) proofs.  The only distinction is that one must account for the fact that the expected welfare might increase by a factor of $2$ due to the expected value of a non-zero coordinate lying in $[1,2]$.  

\begin{lemma}\label{lem:alg.welfare.bic}
  For any valid pair $S$ and $T$, $\Wel(A_{S,T}) \ge  \frac N {4} = \Omega(N)$.
\end{lemma}
\begin{proof}
Let $x$ be an input drawn from the input distribution.
As in the proof of Lemma~\ref{lem:alg.welfare}, we have that $|x| \in [N/2, N]$ and $|x \cap T| \leq \eps_T N$ with probability at least $1/2$.  This implies that, with probability at least $1/2$, $A_{S,T}(x)$ allocates $1$ to each non-zero entry in $x$, of which there are at least $N/2$.  Since $E[x_i | x_i > 0] > 1$ for each $i$, the expected welfare generated by $A_{S,T}$ is at least $N/4$.
\end{proof}

\begin{lemma}\label{lem:mech.T.bic}
  $\Prob_{S,T\sim \Gamma}[\Wel(\Mech_{S,T}, T) > 2\eps_T N] < e^{n^\epsilon} \cdot e^{-(\epsilon_S - \epsilon_{ST})N/4}  = O(e^{-n^{1/4}} )$.
\end{lemma}
\begin{proof}
Choose some valid pair $S$ and $T$.  Since $\E[x_i | x_i > 0] < 2$ for each $i$, the only way for $\Wel(\Mech_{S,T}, T)$ to be
greater than $2 \eps_T N$ is for it to allocate $1$ to at least $\eps_T N$ elements of $T$.  For this to occur on input $T$, the transformation must
query at least one input $x$ such that $|A_{S,T}(x) \cap T| > \eps_T N$.  The proof now follows in exactly the same way as Lemma~\ref{lem:mech.T}.
\end{proof}

We will next prove the following variation of Lemma~\ref{lem:mech.S}.  The statement is the same, but we will use the DSIC constraint rather than the MIDR constraint.

\begin{lemma}\label{lem:bic.mech.S}
Choose $S$ with $|S| = N/2$ and some $x \supset S$ with $|x| \in [N/2, N]$.  Then 
$$\Exp_{T \sim \Gamma_S}[ \Wel(\Mech_{S,T},x) \big\rvert x \cap T = S \cap T] \leq \frac {5 \eps_T N} {\eps_{ST}} = O\left( \frac{N}{n^{1/4}}\right).$$
\end{lemma}
\begin{proof}
As in the proof of Lemma~\ref{lem:mech.S}, we can bound the probability that on input $x$ the mechanism finds a query point $z$ such that $|z| \leq N$ and $A_{S,T}(z) = \emptyset$, say by some $\zeta$ that vanishes faster than any polynomial in $n$.  This part of the proof carries over without change.

Next, we apply the matching characterization of IC, Theorem~\ref{thm:match}, to a pair of types: $x$ and $\alpha \cdot T$.  That is, we will choose $\Domain' = \{x, \alpha \cdot T\}$, and apply Theorem~\ref{thm:match} to the graph $G(\Domain')$.
This yields
$$
\alpha \cdot \Wel(\Mech_{S,T}, T) + \Wel(\Mech_{S,T}, x) \geq \alpha \cdot \Exp_{y \sim \Mech_{S,T}(x)}[ T \cdot y ].
$$
This is because the left hand side is precisely the weight of the identity matching, and the right hand side is the weight of the edge from $\alpha \cdot T$ to $x$.

Taking an expectation over all $T$ such that $S$, $T$ are valid and $x \cap T = S \cap T$, and applying our bound on the likelihood that $\Mech_{S,T}$ is identical to $\Mech_A$, we have
\begin{align*}
\alpha \cdot \Exp_{T\sim \Gamma_S }[\Wel(\Mech_{S,T}, T)] + \Exp_{T\sim \Gamma_S }[\Wel(\Mech_{S,T}, x)]
& \geq \alpha \cdot \Exp_{T\sim \Gamma_S }\left[ \Exp_{y \sim \Mech_{S,T}(x)}[T \cdot y] \right] \\
& \geq \alpha \cdot \frac{\epsilon_{ST} N}{|x|} \cdot \Exp_{T\sim \Gamma_S }[\Wel(\Mech_{S,T}(x), x)] - 2 \zeta N \\
& \geq \alpha \cdot \epsilon_{ST} \cdot \Exp_{T\sim \Gamma_S }[\Wel(\Mech_{S,T}, x)] - 2 \zeta N.
\end{align*}
Rearranging and applying Lemma~\ref{lem:mech.T}, we have
$$
\Wel(\Mech_{S,T}, x) \leq \frac{2 \zeta N}{\alpha \epsilon_{ST} - 1} + \frac{2 \alpha \epsilon_T}{\alpha \epsilon_{ST} - 1} N.
$$
Set $\alpha = 2 \epsilon_{ST}^{-1}$ to get the desired result.
\end{proof}

The remainder of the proof of Theorem~\ref{thm:main.bic} then follows precisely as in the proof of Theorem~\ref{thm:main.midr}.

Recall that for a single agent, BIC and DSIC are equivalent.  Thus, as a corollary of Theorem~\ref{thm:main.bic}, there is no polytime BIC black-box reduction for multiple additive bidders with independently-distributed item values, even for downward-closed feasibility constraints.

\begin{theorem}
\label{thm:main.bic.cor}
  There are constants $\epsilon > 0$ and $c > 0$ such that the following is true in the setting of additive buyers with independently-distributed values for items.
  For any BIC black-box transformation $\Mech$ with query complexity $e^{n^\epsilon}$,
   there exists an algorithm $A$ and distribution $\Dist$ 
  such that $\Wel(\Mech_A) \le \frac{ \Wel(A) }{ n^{c} }$.  This is true even if the feasibility constraint $\Alloc$ is restricted to be downward-closed.
\end{theorem}

\section{Conclusion}
\label{sec:conclusion}

In this paper we presented 
impossibility results for
black-box reductions when agent types are drawn independently from known priors, and the goal is to maximize welfare in expectation over these types.  
We showed that there are no polytime black-box reductions for BIC welfare maximization with additive bidders, even when there is only a single bidder and even for downward-closed feasibility constraints.  This pairs with the existence of BIC black-box reductions for welfare whose query complexity scales with the support size of each agent's type distribution~\cite{HKM15,DHKN2017}, and shows that this dependence is unavoidable.  We also showed the impossibility of black-box reductions for MIDR welfare maximization with $n$ single-dimensional bidders, again in comparison with the BIC case where such reductions are possible.

An important question that remains open is whether there is an $O(1)$-approximate DSIC black-box reduction for expected welfare with single-parameter agents, even when restricting to downward-closed feasibility constraints.  Note that if values lie in $[1,H]$, then an $O(\log H)$ approximation is implied by Babaioff, Lavi, and Pavlov, and this also implies an $O(\log n)$ approximation~\cite{BLP09}.  Can this be improved to a constant approximation?

\subsection*{Acknowledgements}
We thank Shuchi Chawla for many helpful comments.

\bibliographystyle{plain}
\bibliography{bib}

\appendix

\end{document}